\documentclass[12pt,a4paper]{article}

\usepackage{graphicx}
\usepackage{amsmath,amssymb,amsfonts,amsthm}
\usepackage{braket}
\usepackage{array}
\usepackage{comment}
\usepackage[dvipdfmx]{color}
\usepackage{algorithmic,algorithm}
\usepackage{bm}
\usepackage{multirow,array}
\usepackage{url}

\newtheorem{Theorem}{Theorem}
\newtheorem{Definition}[Theorem]{Definition}

\newtheorem{Lemma}[Theorem]{Lemma}

\newlength\savedwidth
\newcommand{\wcline}[1]{\noalign{\global\savedwidth\arrayrulewidth\global\arrayrulewidth 0.9pt} \cline{#1}
\noalign{\global\arrayrulewidth\savedwidth}}

\newlength\myindent
\newcommand\bindent{
     \begingroup
     \setlength{\itemindent}{\myindent}
     \addtolength{\algorithmicindent}{\myindent}
}
\newcommand\eindent{\endgroup}

\begin{document}
\title{Single Quantum Deletion Error-Correcting Codes}

\author{
Ayumu NAKAYAMA
\thanks{
Department of Mathematics and Informatics,
Graduate School of Science and Engineering,
Chiba University
1-33 Yayoi-cho, Inage-ku, Chiba City,
Chiba Pref., JAPAN, 263-0022
}
 \and
Manabu HAGIWARA
\thanks{
Department of Mathematics and Informatics,
Graduate School of Science,
Chiba University
1-33 Yayoi-cho, Inage-ku, Chiba City,
Chiba Pref., JAPAN, 263-0022
}
}

\date{}
\maketitle

\begin{abstract}
In this paper, we discuss a construction method of quantum deletion error-correcting codes.
First of all, we define deletion errors for quantum states, an encoder, a decoder,
and two conditions which is expressed by only the combinatorial language.
Then, we prove that quantum deletion error-correcting codes can be constructed
by two sets that satisfy the conditions.
In other words, problems that correct the deletion errors for quantum states are reduced to
problems that find the sets satisfying the condition by this paper.
Also, we performed experiment of the codes over IBM Quantum Experience.
\end{abstract}

\section{Introduction}
In the classical coding theory, deletion error-correcting codes have been studied
for synchronization error-correction of communication since the pioneer work of Levenshtein in 1966 \cite{levenshtein1966binary}.
These codes have interesting applications such as error-correction for DNA storages \cite{buschmann2013levenshtein},
error-correction for racetrack memories \cite{chee2018coding}, etc.

In the quantum coding theory, there are some well-known codes
such as CSS codes \cite{calderbank1996good}, stabilizer codes \cite{gottesman1997stabilizer},
surface codes \cite{fowler2012surface}, etc. However, these codes cannot correct deletion errors.
Therefore, in this paper, we discuss the construction of quantum deletion error-correcting codes.
Here deletion errors for quantum states are errors corresponding to the case where a part of
photons which a sender transmitted do not reach to a receiver or where quantum states disappear
by the energy decays.

Recently a few papers about quantum deletion error-correcting codes have been published.
In the paper \cite{leahy2019quantum}, quantum insertion-deletion channels are introduced and
in the papers \cite{nakayama2020first} and \cite{hagiwara2020four}, two types of
quantum deletion codes are constructed.

In this paper, we define an encoder and a decoder for given sets $A,B\subset \{0,1\}^n$ and
provide two combinatorial conditions for $A$ and $B$ that make the image of the encoder
single quantum deletion error-correctable.
The codes in the paper \cite{nakayama2020first} and \cite{hagiwara2020four} are the specific instances.

This paper is organized as follows.
In section \ref{Pre}, some notations and definitions are introduced, and
in section \ref{Deletion}, the deletion error for a quantum state is defined.
In section \ref{Construct}, an encoder, a decoder, and two conditions (C1) and (C2) for sets $A,B\subset \{0,1\}^n$ are defined.
In section \ref{Example}, the sets $A,B$ satisfying the conditions (C1) and (C2) are introduced.
In section \ref{Experiment}, a circuit of an encoder and a decoder are created over IBM Quantum Experience \cite{IBM}
that is indicated in \cite{hagiwara2020four} and experiments were performed.
Finally this paper is summarized in section \ref{Summary}.

Partial results have been presented at Japanese domestic workshop in Japanese \cite{nakayama2020construction}.

\section{Preliminaries}\label{Pre}
Let $n$ be an integer greater than or equal to 2 and
$[n]:=\{1,2,\ldots ,n\}$.
For a square matrix $A$ over a complex field $\mathbb{C}$,  $\mathrm{Tr}(A)$ denotes
the sum of the diagonal elements of $A$. 
Set $\ket{0},\ket{1}\in \mathbb{C}^2$ as $\ket{0}:=(1,0)^T,\ket{1}:=(0,1)^T$, and 
$\ket{\bm{x}}$ as $\ket{\bm{x}}:=\ket{x_1}\otimes \ket{x_2}\otimes \cdots \otimes \ket{x_n}$
for a bit sequence $\bm{x}=x_1x_2\cdots x_n\in \{0,1\}^n$. 
Here $\otimes$ is the tensor product operation and $T$ is the transpose operation.
We denote by $S(\mathbb{C}^{2\otimes n})$ the set of all density matrices of order $2^n$.
A density matrix is employed to represent a quantum state. In a particular case $n=1$,
we call $\sigma\in S(\mathbb{C}^2)$ a quantum message.

\begin{Definition}[partial trace $\mathrm{Tr}_i$]
Let $i\in [n]$ be an integer.
Define a function $\mathrm{Tr}_i:S(\mathbb{C}^{2\otimes n})\rightarrow S(\mathbb{C}^{2\otimes n-1})$
as
\begin{align*}
\mathrm{Tr}_i(A):=&\sum_{\bm{x},\bm{y}\in \{0,1\}^n}
a_{\bm{x},\bm{y}}\cdot \mathrm{Tr}(\ket{x_i}\bra{y_i}) \ket{x_1}\bra{y_1}\otimes \\
&\cdots \otimes \ket{x_{i-1}}\bra{y_{i-1}}\otimes \ket{x_{i+1}}\bra{y_{i+1}}\otimes \\
&\cdots \otimes \ket{x_n}\bra{y_n},
\end{align*}
where
$$A=\sum_{\bm{x},\bm{y}\in \{0,1\}^n} a_{\bm{x},\bm{y}}\ket{x_1}\bra{y_1}\otimes \cdots \otimes \ket{x_n}\bra{y_n}.$$
The map $\mathrm{Tr}_i$ is called the partial trace.
\end{Definition}

\begin{Definition}[projection measurement $\mathbb{P}$]
A set $\mathbb{P}:=\{P_1,\ldots P_m\}$ of complex matrices $P_1,\ldots P_m$ of order $2^n$ is called a projection measurement if and only if
every $P_i$ is a projection matrix and
$$\sum_{i=1}^m P_i=\mathbb{I},$$
where $\mathbb{I}$ is the identity matrix of order $2^n$.
If we perform the projection measurement $\mathbb{P}$ under a quantum state $\rho\in S(\mathbb{C}^n)$,
the probability to obtain an outcome $k\in [m]$ is $\mathrm{Tr}(P_k\rho)$,
and the state associated with $k$ after the measurement $\rho '$ is $$\rho ':=\frac{P_k\rho P_k}{\mathrm{Tr}(P_k\rho)}.$$
\end{Definition}

\section{Single Deletion Errors}\label{Deletion}
In this section, we define a single deletion error for a quantum state.
Remark that in the classical coding theory, a single deletion error is defined as a map from a bit sequence
$x_1\dots x_{i-1}x_ix_{i+1}\ldots x_n\in \{0,1\}^n$ to a shorter sequence
$x_1\ldots x_{i-1}x_{i+1}\ldots x_n\in \{0,1\}^{n-1}$ for some $i$.

\begin{Definition}[deletion error $D_i$] \label{deletionError}
Let $i\in [n]$ be an integer.
Define a function $D_i:S(\mathbb{C}^{2\otimes n})\rightarrow S(\mathbb{C}^{2\otimes n-1})$
as $$D_i(\rho):=\mathrm{Tr}_i(\rho),$$
where $\rho\in S(\mathbb{C}^{2\otimes n})$ is a quantum state.
We call the map $D_i$ a deletion error.
\end{Definition}

\section{Construction of Quantum Deletion Error-Correcting Codes}\label{Construct}
In this section, we define an encoder and a decoder and
propose two conditions for
two sets of bit sequences. Afterwards, we discuss the construction
of the quantum deletion error-correcting codes.

\begin{Definition}[encoder $\mathrm{Enc}_{A,B}$]
Let $A,B\subset \{0,1\}^n$ be sets which satisfy the following conditions
$A\neq \emptyset,B\neq \emptyset,A\cap B=\emptyset$.
Define an encoder $\mathrm{Enc}_{A,B}:S(\mathbb{C}^2)\rightarrow S(\mathbb{C}^{2\otimes n})$
as $$\mathrm{Enc}_{A,B}(\sigma):=\ket{\Psi}\bra{\Psi},$$
where $\sigma:=\ket{\psi}\bra{\psi}\in S(\mathbb{C}^2)$ is a quantum message
with a unit vector $\ket{\psi}:=\alpha\ket{0}+\beta\ket{1}\in\mathbb{C}^2$ and
$$\ket{\Psi}:=\frac{\alpha}{\sqrt{|A|}}\sum_{\bm{a}\in A}\ket{a}
+\frac{\beta}{\sqrt{|B|}}\sum_{\bm{b}\in B}\ket{b}.$$
\end{Definition}

\begin{Definition}[$\Delta_{i,b}$]
Let $i\in [n]$ and $b\in \{0,1\}$ be integers.
Define a set $\Delta_{i,b}(A)\subset \{0,1\}^{n-1}$ as
\begin{align*}
\Delta_{i,b}(A):=\{(a_1,\ldots ,a_{i-1},a_{i+1},\ldots ,a_n)\in \{0,1\}^{n-1}|\\
(a_1,\ldots ,a_{i-1},b,a_{i+1},\ldots ,a_n)\in A\},
\end{align*}
where $A\subset \{0,1\}^n$ is a non-empty set.
We call the set $\Delta_{i,b}(A)$ an $(i,b)$ deletion set of $A$.
\end{Definition}

\begin{Definition}[conditions (C1) and (C2)]
For non-empty sets $A,B\subset \{0,1\}^n$,
define two conditions (C1) and (C2) as follows.
\begin{itemize}
\item[](C1 distance condition):\\
For any $i_1,i_2\in [n]$ and any $b_1,b_2\in \{0,1\}$,
\begin{equation*}
|\Delta_{i_1,b_1}(A)\cap \Delta_{i_2,b_2}(B)|=0.
\end{equation*}
\item[](C2 ratio condition):\\
For any $i_1,i_2\in [n]$ and any $b\in \{0,1\}$,
\begin{equation*}
|A||\Delta_{i_1,b}(B)\cap \Delta_{i_2,b}(B)|=|B||\Delta_{i_1,b}(A)\cap \Delta_{i_2,b}(A)|.
\end{equation*}
\end{itemize}
\end{Definition}

\begin{algorithm}[t]
\caption{}
\begin{algorithmic}
[1]\STATE Input: $A,B\subset \{0,1\}^n$. Output: $\mathbb{P}=\{P_1,\ldots ,P_j\}$.
\STATE $j:=1,\mathbb{P}:=\emptyset$.
\bindent
\WHILE{there exist $i_1\in [n],i_2\in [n] \backslash \{i_1\},b\in \{0,1\}$\\
\hfill such that $\Delta_{i_1,b}(A)\cap \Delta_{i_2,b}(A)\neq \emptyset$}
\STATE take $a_0+b_0$ distinct elements\\
$\bm x_1,\bm x_2,\ldots ,\bm x_{a_0}\in \Delta_{i_1,b}(A)\cap \Delta_{i_2,b}(A)$,\\ 
$\bm y_1,\bm y_2,\ldots ,\bm y_{b_0}\in \Delta_{i_1,b}(B)\cap \Delta_{i_2,b}(B)$,\\
and define the followings:\\
$P_j:=\sum_{t=1}^{a_0}\ket{\bm x_t}\bra{\bm x_t}+\sum_{s=1}^{b_0}\ket{\bm y_s}\bra{\bm y_s}$,\\
$\mathbb{P}:=\mathbb{P}\cup \{P_j\}$,\\
$j:=j+1$,\\
for all $i\in [n]$,\\
\hfil $\Delta_{i,b}(A):=\Delta_{i,b}(A)\backslash \{\bm x_1,\ldots ,\bm x_{a_0}\}$,\\
$\hfil \Delta_{i,b}(B):=\Delta_{i,b}(B)\backslash \{\bm y_1,\ldots ,\bm y_{b_0}\}$.
\ENDWHILE
\WHILE{there exist $i\in [n],b\in \{0,1\}$\\
\hfill such that $\Delta_{i,b}(A)\neq \emptyset$}
\STATE take $a_0+b_0$ distinct elements\\
$\bm x_1,\bm x_2,\ldots ,\bm x_{a_0}\in \Delta_{i,b}(A)$,\\ 
$\bm y_1,\bm y_2,\ldots ,\bm y_{b_0}\in \Delta_{i,b}(B)$,\\
and define the followings:\\
$P_j:=\sum_{t=1}^{a_0}\ket{\bm x_t}\bra{\bm x_t}+\sum_{s=1}^{b_0}\ket{\bm y_s}\bra{\bm y_s}$,\\
$\mathbb{P}:=\mathbb{P}\cup \{P_j\}$,\\
$j:=j+1$,\\
$\Delta_{i,b}(A):=\Delta_{i,b}(A)\backslash \{\bm x_1,\ldots ,\bm x_{a_0}\}$,\\
$\Delta_{i,b}(B):=\Delta_{i,b}(B)\backslash \{\bm y_1,\ldots ,\bm y_{b_0}\}$.
\ENDWHILE
\STATE $P_j:=\mathbb{I}-\sum_{P\in \mathbb{P}}P$.\nonumber \\
\STATE $\mathbb{P}:=\mathbb{P}\cup \{P_j\}$.
\eindent
\RETURN $\mathbb{P}$
\end{algorithmic}
\end{algorithm}

\begin{Lemma}
Let $A,B\in \{0,1\}^n$ be non-empty sets satisfying the conditions (C1) and (C2).
Define $(a_0,b_0):=(\frac{|A|}{\mathrm{gcd}(|A|,|B|)},\frac{|B|}{\mathrm{gcd}(|A|,|B|)})$.
Then the output $\mathbb{P}$ of Algorithm 1 is a projection measurement.
Here $\mathrm{gcd}(|A|,|B|)$ is the greatest common divisor of $|A|$ and $|B|$.
\end{Lemma}

\begin{proof}
By the condition (C2), $\Delta_{i_1,b}(B)\cap \Delta_{i_2,b}(B)\neq \emptyset$ holds as
$\Delta_{i_1,b}(A)\cap \Delta_{i_2,b}(A)\neq \emptyset$ does.
Then we obtain a ratio $|A|:|B|=|\Delta_{i_1,b}(A)\cap \Delta_{i_2,b}(A)|:|\Delta_{i_1,b}(B)\cap \Delta_{i_2,b}(B)|=a_0:b_0$.
Therefore we are able to take $a_0+b_0$ distinct elements
$$\bm x_1,\bm x_2,\ldots ,\bm x_{a_0}\in \Delta_{i_1,b}(A)\cap \Delta_{i_2,b}(A)$$
$$\bm y_1,\bm y_2,\ldots ,\bm y_{b_0}\in \Delta_{i_1,b}(B)\cap \Delta_{i_2,b}(B)$$
if the step at 2: in Algorithm 1 holds.
Similarly, we can take $a_0+b_0$ distinct elements
$$\bm x_1,\bm x_2,\ldots ,\bm x_{a_0}\in \Delta_{i,b}(A),\bm y_1,\bm y_2,\ldots ,\bm y_{b_0}\in \Delta_{i,b}(B)$$
if the step at 5: in Algorithm 1 holds.
Note that $\braket{\bm x_t|\bm y_s}=0$ follows from the condition (C1).

Hence, by the definitions of $P_1,\ldots ,P_{p-1}$,
it is clear that these matrices are projection matrices, where
$p:=|\mathbb{P}|$.
It is also trivial that $\sum_{P\in \mathbb{P}}P=\mathbb{I}$ holds due to  the definition of $P_p$.
 Let
 \begin{equation*}
P:=\sum_{t=1}^{a_0}\ket{\bm x_t}\bra{\bm x_t}
+\sum_{s=1}^{b_0}\ket{\bm y_s}\bra{\bm y_s},
\end{equation*}
\begin{equation*}
\tilde{P}:=\sum_{t'=1}^{a_0}\ket{\tilde{\bm x}_{t'}}\bra{\tilde{\bm x}_{t'}}
+\sum_{s'=1}^{b_0}\ket{\tilde{\bm y}_{s'}}\bra{\tilde{\bm y}_{s'}}.
\end{equation*}
Then
$$P_p^\dagger=\mathbb{I}^\dagger-\sum_{P}P^\dagger=\mathbb{I}-\sum_{P}P=P_p$$
holds, and
\begin{center}
\begin{align*}
P_p^2&=\mathbb{I}-2\sum_{P}P+\sum_{P,\tilde{P}}P\tilde{P}\\
&=\mathbb{I}-2\sum_{P}P+\sum_{P,\tilde{P}}\delta_{P,\tilde{P}}P\\
&=\mathbb{I}-2\sum_{P}P+\sum_{P}P=P_p
\end{align*}
\end{center}
holds. Therefore $\mathbb{P}$ is a projection measurement.
\end{proof}

\begin{Definition}[error-correcting operator $U_k$] \label{errorCorrectingOperator}
For non-empty sets $A,B\subset \{0,1\}^n$ satisfying the conditions (C1) and (C2)
and a projection measurement $\mathbb{P}$ that is constructed by Algorithm 1,
we can choose a unitary matrix $U_k$ which satisfies the followings for $1\le k\le |\mathbb{P}|$:
$$
U_k\frac{\sum_{\bm x\in \Delta (A)}P_k\ket{\bm x}}{||\sum_{\bm x\in \Delta (A)}P_k\ket{\bm x}||}=\ket{0\cdots 00},
$$
$$
U_k\frac{\sum_{\bm y\in \Delta (B)}P_k\ket{\bm y}}{||\sum_{\bm y\in \Delta (B)}P_k\ket{\bm y}||}=\ket{0\cdots 01},
$$
where
$$\Delta (A):=\bigcup_{\substack{i\in [n]\\b\in \{0,1\}}}\Delta _{i,b}(A),
\Delta (B):=\bigcup_{\substack{i\in [n]\\b\in \{0,1\}}}\Delta _{i,b}(B).$$
We call the matrix $U_k$ an error-correcting operator.
\end{Definition}

\begin{Definition}[decoding algorithm $\mathrm{Dec}_{A,B}$] \label{decoding}
Let $A,B\subset \{0,1\}^n$ be non-empty sets satisfying the conditions (C1) and (C2)
and $\mathbb{P}$ a projection measurement which is constructed by Algorithm 1.
Define a function $\mathrm{Dec}_{A,B}:S(\mathbb{C}^{2\otimes n-1})\rightarrow S(\mathbb{C}^2)$ as
a map which assigns $\rho '\in S(\mathbb{C}^{2\otimes n-1})$ to $\sigma '\in S(\mathbb{C}^2)$ which
is constructed by the following procedure.

\begin{enumerate}
\item Perform the projection measurement $\mathbb{P}$ under the state $\rho '$.
Assume that the outcome is $1\le k\le |\mathbb{P}|$ and that the state after measurement is $\rho '_k$.
\item Let $\tilde{\rho}:=U_k\rho '_k U^\dagger_k$.
Here $U_k$ is the error-correcting operator.
\item At last, return $\sigma ':=\underbrace{\mathrm{Tr}_1\circ \cdots \circ \mathrm{Tr}_1}_{n-2\ \mathrm{times}}(\tilde{\rho})$.
\end{enumerate}
\end{Definition}

\begin{Theorem}
Let $A,B\subset \{0,1\}^n$ be non-empty sets satisfying the conditions (C1) and (C2).
Then for any quantum message $\sigma \in S(\mathbb{C}^2)$ and any deletion position $i\in [n]$,
$$\mathrm{Dec}_{A,B}\circ D_i\circ \mathrm{Enc}_{A,B}(\sigma)=\sigma.$$
Here the symbol $\circ$ represents the composition of functions.
In other words, $\mathrm{Enc}_{A,B}(S(\mathbb{C}^2))$ is a single quantum deletion error-correcting code
with the decoder $\mathrm{Dec}_{A,B}$.
\end{Theorem}

\begin{proof}
For a quantum message $\sigma :=\ket{\psi}\bra{\psi}$ with a unit vector $\ket{\psi}=\alpha \ket{0}+\beta \ket{1}\in \mathbb{C}^2$,
set $\rho$ as
\begin{equation}
\rho :=\mathrm{Enc}_{A,B}(\sigma). \label{rho}
\end{equation}

Then for any $i\in [n]$, $D_i(\rho)$ is rewritten as
\begin{eqnarray*}
D_i(\rho)=\sum_{b\in \{0,1\}}{\Biggl (}
&&\frac{|\alpha|^2}{|A|}\sum_{\substack{\bm x\in \Delta_{i,b}(A)\\ \tilde{\bm x}\in \Delta_{i,b}(A)}}\ket{\bm x}\bra{\tilde{\bm x}}\\
&+&\frac{\alpha \overline{\beta}}{\sqrt{|A||B|}}\sum_{\substack{\bm x\in \Delta_{i,b}(A)\\ \bm y\in \Delta_{i,b}(B)}}\ket{\bm x}\bra{\bm y}\\
&+&\frac{\overline{\alpha}\beta}{\sqrt{|A||B|}}\sum_{\substack{\bm y\in \Delta_{i,b}(B)\\ \bm x\in \Delta_{i,b}(A)}}\ket{\bm y}\bra{\bm x}\\
&&+\frac{|\beta|^2}{|B|}\sum_{\substack{\bm y\in \Delta_{i,b}(B)\\ \tilde{\bm y}\in \Delta_{i,b}(B)}}\ket{\bm y}\bra{\tilde{\bm y}}{\Biggr )}.
\end{eqnarray*}

Let $\mathbb{P}=\{P_1,P_2,\ldots ,P_p\}$ be the projection measurement constructed by Algorithm 1.
By the definition of $P_p$,
$$P_pD_i(\rho)=O$$
for any $i\in [n]$, where $O$ is the zero matrix. In addition,
for any $i\in [n]$ and any $k\in [p-1]$,
the following equation holds.
Here $P_k$ is denoted as
$$P_k=\sum_{t=1}^{a_0}\ket{\bm x^{(k)}_t}\bra{\bm x^{(k)}_t}+\sum_{s=1}^{b_0}\ket{\bm y^{(k)}_s}\bra{\bm y^{(k)}_s}.$$

\begin{align*}
&P_kD_i(\rho)=\\
&\sum_{b\in \{0,1\}}{\Biggl (}\frac{|\alpha|^2}{|A|}\sum_{t=1}^{a_0}\sum_{\substack{\bm x\in \Delta_{i,b}(A)\\ \tilde{\bm x}\in \Delta_{i,b}(A)}}
\braket{\bm x^{(k)}_t|\bm x}\ket{\bm x^{(k)}_t}\bra{\tilde{\bm x}}\\
&+\frac{\alpha \overline{\beta}}{\sqrt{|A||B|}}\sum_{t=1}^{a_0}\sum_{\substack{\bm x\in \Delta_{i,b}(A)\\ \bm y\in \Delta_{i,b}(B)}}
\braket{\bm x^{(k)}_t|\bm x}\ket{\bm x^{(k)}_t}\bra{\bm y}\\
&+\frac{\overline{\alpha}\beta}{\sqrt{|A||B|}}\sum_{s=1}^{b_0}\sum_{\substack{\bm y\in \Delta_{i,b}(B)\\ \bm x\in \Delta_{i,b}(A)}}
\braket{\bm y^{(k)}_s|\bm y}\ket{\bm y^{(k)}_s}\bra{\bm x}\\
&+\frac{|\beta|^2}{|B|}\sum_{s=1}^{b_0}\sum_{\substack{\bm y\in \Delta_{i,b}(B)\\ \tilde{\bm y}\in \Delta_{i,b}(B)}}
\braket{\bm y^{(k)}_s|\bm y}\ket{\bm y^{(k)}_s}\bra{\tilde{\bm y}}{\Biggr )}.
\end{align*}

Therefore in the case where $k=p$,
$$\mathrm{Tr}(P_kD_i(\rho))=0$$
holds for any $i\in [n]$ and in the case where $k\in [p-1]$,
there exists $b\in \{0,1\}$ such that
\begin{align*}
\mathrm{Tr}(P_kD_i(\rho))=
& \frac{|\alpha|^2}{|A|}\sum_{t=1}^{a_0}\sum_{\bm x\in  \Delta_{i,b}(A)}|\braket{\bm x^{(k)}_t|\bm x}|^2\\
& + \frac{|\beta|^2}{|B|}\sum_{s=1}^{b_0}\sum_{\bm y\in  \Delta_{i,b}(B)}|\braket{\bm y^{(k)}_s|\bm y}|^2
\end{align*}
 holds for any $i\in [n]$.
 Hence for $k\in [p-1]$ such that $\mathrm{Tr}(P_kD_i(\rho))\neq 0$,
 $$
\mathrm{Tr}(P_kD_i(\rho))
=\frac{a_0|\alpha|^2}{|A|}+\frac{b_0|\beta|^2}{|B|}
=\frac{1}{\lambda}
$$
and the state after the measurement is
\begin{eqnarray*} 
\lambda P_kD_i(\rho)P_k=
&&\frac{|\alpha|^2}{a_0}\sum_{t=1}^{a_0}\sum_{t'=1}^{a_0}\ket{\bm x^{(k)}_t}\bra{\bm x^{(k)}_{t'}}\\
&+&\frac{\alpha \overline{\beta}}{\sqrt{a_0b_0}}\sum_{t=1}^{a_0}\sum_{s=1}^{b_0}\ket{\bm x^{(k)}_t}\bra{\bm y^{(k)}_s}\\
&+&\frac{\alpha \overline{\beta}}{\sqrt{a_0b_0}}\sum_{s=1}^{b_0}\sum_{t=1}^{a_0}\ket{\bm y^{(k)}_s}\bra{\bm x^{(k)}_t}\\
&+&\frac{|\beta|^2}{b_0}\sum_{s=1}^{b_0}\sum_{s'=1}^{b_0}\ket{\bm y^{(k)}_s}\bra{\bm y^{(k)}_{s'}}\\
&=&\ket{\Phi _k}\bra{\Phi _k},
\end{eqnarray*}
where $\lambda:=\mathrm{gcd}(|A|,|B|)$ and
$$\ket{\Phi _k}:=\frac{\alpha}{\sqrt{a_0}}\sum_{t=1}^{a_0}\ket{\bm x^{(k)}_t}
+\frac{\beta}{\sqrt{b_0}}\sum_{s=1}^{b_0}\ket{\bm y^{(k)}_s}.$$

Note that $\ket{\Phi _k}$ can be rewritten as
$$\ket{\Phi _k}=\alpha \frac{\sum_{\bm x\in \Delta (A)}P_k\ket{\bm x}}{||\sum_{\bm x\in \Delta (A)}P_k\ket{\bm x}||}
+\beta \frac{\sum_{\bm y\in \Delta (B)}P_k\ket{\bm y}}{||\sum_{\bm y\in \Delta (B)}P_k\ket{\bm y}||}.$$
From the discussion above, for any quantum message $\sigma \in S(\mathbb{C}^2)$ and any deletion position $i\in [n]$,
\begin{align*}
&\mathrm{Dec}_{A,B}\circ D_i\circ \mathrm{Enc}_{A,B}(\sigma)\\
&=\mathrm{Dec}_{A,B}\circ D_i(\rho) & (\because \ \eqref{rho})\\
&=\mathrm{Tr}_1\circ \cdots \circ \mathrm{Tr}_1(U_k\ket{\Phi _k}\bra{\Phi _k}U^\dagger _k) & (\because \ Def. \ \ref{decoding})\\
&=\mathrm{Tr}_1\circ \cdots \circ \mathrm{Tr}_1(\ket{0}\bra{0}\otimes \cdots \otimes \ket{0}\bra{0} \otimes \sigma) & (\because \ Def. \ \ref{errorCorrectingOperator})\\
&=\sigma & (\because Def. \ \ref{deletionError})
\end{align*}
holds.
\end{proof}

\section{Examples of Quantum Deletion Error-Correcting Codes}\label{Example}
In this section, we introduce two instances of sets $A,B\subset \{0,1\}^n$
that satisfy the conditions (C1) and (C2).

\subsection{Example 1 (4 qubits code \cite{hagiwara2020four})}
Let $n:=4$,
$$A:=\{0000,1111\},$$
and
$$B:=\{0011,0101,1001,0110,1010,1100\}.$$
Then for any $i\in [4]$,
$$\Delta _{i,0}(A)=\{000\}, \Delta _{i,1}(A)=\{111\},$$
$$\Delta _{i,0}(B)=\{011,101,110\},\Delta _{i,1}(B)=\{001,010,100\}.$$
Therefore for any $i_1,i_2\in [4]$ and any $b_1,b_2\in \{0,1\}$,
$$|\Delta _{i_1,b_1}(A)\cap \Delta _{i_2,b_2}(B)|=0,$$
thus $A$ and $B$ satisfy the condition (C1).
Moreover for any $i_1,i_2\in [4]$ and $b\in \{0,1\}$,
$$|A||\Delta_{i_1,b}(B)\cap \Delta_{i_2,b}(B)|=2\times 3=6,$$
$$|B||\Delta_{i_1,b}(A)\cap \Delta_{i_2,b}(A)|=6\times 1=6.$$
Consequently $A$ and $B$ satisfy the condition (C2).

This code has two interesting properties.
One is that the cardinalities of $A$ and $B$ are different.
Thereby the code can be expressed by neither any CSS  codes \cite{calderbank1996good}
nor any stabilizer codes \cite{gottesman1997stabilizer}.
The other is that the length of the code is 4.
In the paper \cite{hagiwara2020four},
it is proved that the minimum length of the code capable of correcting any single deletion errors is 4.

\subsection{Example 2 (8 qubits code \cite{nakayama2020first})}
Let $n:=8$,
$$A:=\{00001001,01101111\},B:=\{00001111, 01101001\}.$$
The tables \ref{table1} and \ref{table2} show $\Delta _{i,b}(A)$ and $\Delta _{i,b}(B)$
for all $i\in [8]$ and $b\in \{0,1\}$.
By these tables, it is easy to see that $A$ and $B$ satisfy
the conditions (C1) and (C2).

The interesting point of this code is that
the states after single deletion errors may be different depending on the deletion positions.
Furthermore the different states are not even orthogonal (e.g., $D_1(\mathrm{Enc}_{A,B}(\sigma))$ and
$D_2(\mathrm{Enc}_{A,B}(\sigma))$ are not orthogonal).
However, this code is also a quantum deletion error-correcting code.

\begin{figure*}[h]
\begin{center}
\includegraphics[scale=0.36]{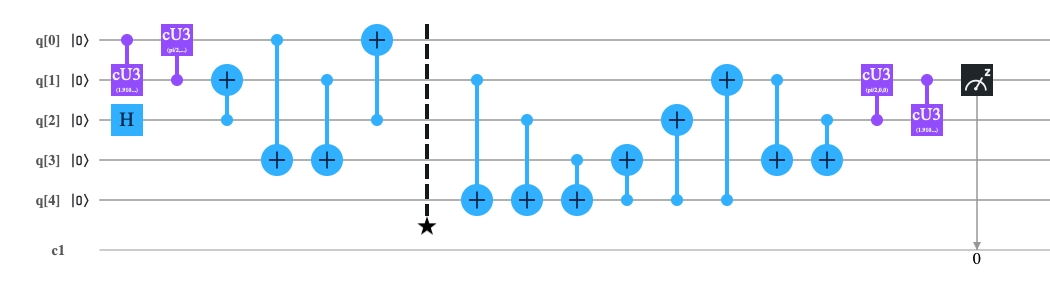}
\caption{encoder and decoder circuit}
\label{Circuit}
\end{center}
\end{figure*}

\begin{table}[t]
\centering
\newcolumntype{I}{!{\vrule width 0.9pt}}
\begin{tabular}{|cIc|c|}
\hline
\multirow{2}{*}{$\Delta_{i,b}(A)$} & \multicolumn{2}{c|}{$A=\{00001001,01101111\}$} \\ \cline{2-3} 
                              & $b=0$                  & $b=1$                 \\ \wcline{1-3}
$i=1$                         & $\{0001001,1101111\}$  & $\emptyset$           \\ \hline
$i=2$                         & $\{0001001\}$          & $\{0101111\}$         \\ \hline
$i=3$                         & $\{0001001\}$          & $\{0101111\}$         \\ \hline
$i=4$                         & $\{0001001,0111111\}$  & $\emptyset$           \\ \hline
$i=5$                         & $\emptyset$            & $\{0000001,0110111\}$ \\ \hline
$i=6$                         & $\{0000101\}$          & $\{0110111\}$         \\ \hline
$i=7$                         & $\{0000101\}$          & $\{0110111\}$         \\ \hline
$i=8$                         & $\emptyset$            & $\{0000100,0110111\}$ \\ \hline
\end{tabular}
\caption{$\Delta_{i,b}(A)$ for $i\in [8], b \in \{0,1\}$}
\label{table1}
\newcolumntype{I}{!{\vrule width 0.9pt}}
\begin{tabular}{|cIc|c|}
\hline
\multirow{2}{*}{$\Delta_{i,b}(B)$} & \multicolumn{2}{c|}{$B=\{00001111, 01101001\}$} \\ \cline{2-3} 
                                 & $b=0$                  & $b=1$                  \\ \wcline{1-3}
$i=1$                            & $\{0001111,1101001\}$  & $\emptyset$            \\ \hline
$i=2$                            & $\{0001111\}$          & $\{0101001\}$          \\ \hline
$i=3$                            & $\{0001111\}$          & $\{0101001\}$          \\ \hline
$i=4$                            & $\{0001111,0111001\}$  & $\emptyset$            \\ \hline
$i=5$                            & $\emptyset$            & $\{0000111,0110001\}$  \\ \hline
$i=6$                            & $\{0110101\}$          & $\{0000111\}$          \\ \hline
$i=7$                            & $\{0110101\}$          & $\{0000111\}$          \\ \hline
$i=8$                            & $\emptyset$            & $\{0000111,0110100\}$  \\ \hline
\end{tabular}
\caption{$\Delta_{i,b}(A)$ for $i\in [8], b \in \{0,1\}$}
\label{table2}
\newcolumntype{I}{!{\vrule width 0.9pt}}
\scalebox{0.97}[1.0]{
\begin{tabular}{|cIc|c|c|c|}
\hline
  & \multicolumn{2}{c|}{Simulator} & \multicolumn{2}{c|}{Quantum computer} \\ \hline
Initial state & Outcome:0 & Outcome:1 & Outcome:0 & Outcome:1 \\ \wcline{1-5}
$\ket{0}$ & $100\%$ & $0\%$ & $69.922\%$ & $30.078\%$  \\ \hline
$\ket{1}$ & $0\%$ & $100\%$ & $35.669\%$ & $64.331\%$ \\ \hline
$\frac{1}{\sqrt{2}}(\ket{0}+\ket{1})$ & $50.562\%$ & $49.438\%$ & $54.565\%$ & $45.435\%$ \\ \hline
\end{tabular}
}\caption{Theoretical values and experimental values for three types of initial states}
\label{table3}
\end{table}

\section{Experiment by IBM Quantum Experience} \label{Experiment}
\subsection{About IBM Quantum Experience}
IBM Quantum Experience (IBM Q) \cite{IBM} is a service to implement circuits which users create,
does experiments on the quantum computer IBM possesses, and outputs the results to the users. 

\subsection{Implementation over IBM Quantum Experience}
In the paper \cite{hagiwara2020four}, circuits of an encoder and two decoders are
constructed. 
We implemented the encoder and the decoder on IBM Q and 
did experiments to see if the quantum deletion error-correcting codes work on the real quantum computer.
The figure \ref{Circuit} shows the circuit on IBM Q.
The circuit before the symbol $\star$ in the figure represents the encoder circuit,
the circuit after is the decoding circuit,
and we express the deletion error $D_1$ by not using the first qubit right after encoding.

\subsection{Experimental Results}
The table \ref{table3} shows the percentages that the outcomes $0$ and $1$ obtained
on the simulation and the quantum computer
for three types of states $\ket{0}$, $\ket{1}$, and $\frac{1}{\sqrt{2}}(\ket{0}+\ket{1})$.
By the simulation results, it is proved that the circuit we created on IBM Q works correctly.
However, the results by the quantum computer indicates that
error-correcting fails with $30\sim 35\%$.

Qubits in the quantum computer on IBM Q are arranged in a two-dimensional form
and it is not possible to operate two nonadjacent qubits simultaneously.
Therefore operations of two nonadjacent qubits are transformed into
combinations of operations of two adjacent qubits. 
It causes the increase of the number of total quantum gates.
It is known that the more the gates increase, the more errors occur.
Hence the future work is to redesign the circuit to reduce the number of gates.

\section{Summary}\label{Summary}
In this paper, we discussed the quantum deletion error-correcting codes.
First of all, we defined the deletion errors for quantum states by using the partial trace operations.
Since taking partial trace for a density matrix corresponds to expressing the state of its subsystem,
it is natural that a deletion error is defined by using the partial trace.

In section \ref{Construct}, we defined the encoder, decoder, and the conditions (C1) and (C2)
for sets $A,B\subset \{0,1\}^n$, and proved that quantum error-correcting codes for single deletion errors
can be constructed by employing the sets $A,B\subset \{0,1\}^n$ which satisfy the conditions (C1) and (C2).
In addition, we introduced two kinds of examples satisfying the conditions (C1) and (C2) in section \ref{Example}.

In section \ref{Experiment}, we created the encoder and decoder circuit on IBM Q
which is shown in the paper \cite{hagiwara2020four}.
As the experimental results, we found that
there are $30\sim 35\%$ miscorrections on the quantum computer.
This is due to the increase of quantum gates so that redesigning the circuit is the future work.

\section*{Acknowledgments}
This paper is partially supported by
KAKENHI 18H01435.

\bibliographystyle{plain}
\bibliography{bibtex}

\begin{thebibliography}{10}

\bibitem{IBM}
\url{https://quantum-computing.ibm.com}.

\bibitem{buschmann2013levenshtein}
Tilo Buschmann and Leonid~V Bystrykh.
\newblock Levenshtein error-correcting barcodes for multiplexed dna sequencing.
\newblock {\em BMC bioinformatics}, 14(1):272, 2013.

\bibitem{calderbank1996good}
A~Robert Calderbank and Peter~W Shor.
\newblock Good quantum error-correcting codes exist.
\newblock {\em Physical Review A}, 54(2):1098--1105, 1996.

\bibitem{chee2018coding}
Yeow~Meng Chee, Han~Mao Kiah, Alexander Vardy, Eitan Yaakobi, et~al.
\newblock Coding for racetrack memories.
\newblock {\em IEEE Transactions on Information Theory}, 64(11):7094--7112,
  2018.

\bibitem{fowler2012surface}
Austin~G Fowler, Matteo Mariantoni, John~M Martinis, and Andrew~N Cleland.
\newblock Surface codes: Towards practical large-scale quantum computation.
\newblock {\em Physical Review A}, 86(3):032324, 2012.

\bibitem{gottesman1997stabilizer}
Daniel~Eric Gottesman.
\newblock {\em Stabilizer Codes and Quantum Error Correction}.
\newblock PhD thesis, California Institute of Technology, 1997.

\bibitem{hagiwara2020four}
Manabu Hagiwara and Ayumu Nakayama.
\newblock A four-qubits code that is a quantum deletion error-correcting code
  with the optimal length.
\newblock {\em arXiv preprint arXiv:2001.08405}, 2020.

\bibitem{leahy2019quantum}
Janet Leahy, Dave Touchette, and Penghui Yao.
\newblock Quantum insertion-deletion channels.
\newblock {\em ArXiv}, abs/1901.00984, 2019.

\bibitem{levenshtein1966binary}
Vladimir~I Levenshtein.
\newblock Binary codes capable of correcting deletions, insertions, and
  reversals.
\newblock In {\em Soviet physics doklady}, volume~10, pages 707--710, 1966.

\bibitem{nakayama2020construction}
Ayumu Nakayama and Manabu Hagiwara.
\newblock Construction of quantum error correcting codes for single deletion
  errors (in japanese with english summary).
\newblock {\em IEICE Technical Report}, 119(376, IT2019-68):185--189, 2020.

\bibitem{nakayama2020first}
Ayumu Nakayama and Manabu Hagiwara.
\newblock The first quantum error-correcting code for single deletion errors.
\newblock {\em IEICE Communications Express}, page 2019XBL0154, 2020.

\end{thebibliography}
\end{document}